\documentclass[twocolumn]{svjour3}                     
\smartqed  
\usepackage{graphicx}
\usepackage{algorithmic}
\usepackage{algorithm}
\usepackage{epsfig}

\begin{document}

\title{Hypo-Steiner Heuristic for Multicast Routing in All-Optical WDM Mesh Networks}


\subtitle{} 


\author{Fen Zhou \and Mikl\a'os Moln\a'ar \and Bernard Cousin
}


\institute{
              Fen Zhou, Mikl\a'os Moln\a'ar, and Bernard Cousin\at
              IRISA, Campus de Beaulieu, Rennes 35042, France   \\
              Tel.: +33-299847194\\
              Fax: +33-299847171\\
              \email{\{fen.zhou, molnar, bernard.cousin\}@irisa.fr}
             \and
             Fen Zhou, and Mikl\a'os Moln\a'ar \at
             INSA Rennes, 20, Avenue des Buttes de Co\"{e}smes, Rennes 35042, France
             \and
             Bernard Cousin \at
             University of Rennes 1, Campus de Beaulieu, Rennes 35042, France
}

\date{Received: date / Accepted: date}

\maketitle

\begin{abstract}
In sparse light splitting all-optical WDM networks, the more destinations a light-tree can accommodate, the fewer light-trees and wavelengths a multicast session will require. In this article, a Hypo-Steiner Light-tree algorithm (HSLT) is proposed to construct a HSLT light-tree to include as many destinations as possible. The upper bound cost of the light-trees built by HSLT is given as $N(N-1)/2$, where $N$ is the number of nodes in the network. The analytical model proves that, under the same condition, more destinations could be held in a HSLT light-tree than a Member-Only~\cite{xjzhang2000} light-tree. Extensive simulations not only validate the proof but also show that the proposed heuristic outperforms the existing multicast routing algorithms by a large margin in terms of link stress, throughput, and efficiency of wavelength usage.
\end{abstract}

\section{Introduction}
\label{Introduction}
With the increasing popularity of bandwidth-driven and time sensitive applications (like Video-Conference, HTDV and Video On Demand), all-optical WDM networks for the Next Generation Internet (NGI) need to provide multicast services with high throughput and small latency~\cite{jsTurner2002}. In order to support multicast in WDM networks, a light-tree concept was proposed in~\cite{lhSahasrabuddhe1999}, which is a tree in the physical topology and occupies the same wavelength over all the fiber links in the tree. It has been found to provide substantial savings in the network-wide average packet hop distance and the total number of transceivers in the networks~\cite{bMukherjee2000}. However, due to the physical constraints and characteristics in all-optical WDM networks, multicast routing is a challenging work. First, in the absence of any wavelength conversion device, the same wavelength should be employed over the light-tree, which is referred as the wavelength continuity constraint~\cite{bMukherjee2000}. Meanwhile, two or more light-trees traversing the same fiber link must be assigned different wavelengths, so that they do not interfere with one another, which is referred as the distinct wavelength constraint~\cite{bMukherjee2000}. In addition, since all-optical multicast has to distribute packets in the optical domain, branching nodes (or switch nodes) in a light-tree is required to be equipped with light splitters. By employing the light splitting capability, the branching node is able to replicate the incoming packets in the optical domain and forward them to all the required outgoing ports. Usually, a node capable of light splitting is named as a Multicast Capable node (MC)~\cite{rMalli1998}. Otherwise, it is a Multicast Incapable node (MI)~\cite{rMalli1998}. Typically, the network nodes at least possess of the Tap and Continue (TaC~\cite{mAli2000Cost}) capability to tap into the light signal for local consumption and forward it to only one outgoing port. From the point of optical energy budget, a light splitter reduces the power level of a light signal by a factor equal to the number of optical copies. The reduction of power should be compensated by internal active amplifiers like erbium-doped fiber amplifier (EDFA)~\cite{eDesurvire1991}, which, however, introduce many problems such as Gain Dispersion, Gain Saturation and Noise~\cite{sgYan2003}. Consequently, the complex architectures along with the high-cost of optical amplification make MC nodes much more expensive than the MI nodes. That is why only a subset of nodes in the WDM networks support light splitting, which is characterized as sparse light splitting~\cite{rMalli1998}.

Multicast routing in sparse light splitting WDM networks mainly involves two subproblems: light-tree construction and wavelength assignment. The first subproblem is to build a light-tree to include destinations in a multicast session under the degree limitation of MI nodes. Since the MI nodes cannot be a branching node in a light-tree, one light-tree may not be enough to include all the destinations in the multicast session and several ones may be required, which thus forms a light-forest~\cite{xjzhang2000}. In a light-tree, generally, one emitter of the source node transmits light signals on only one wavelength to the spanned destinations. In case that the source node can not split, several emitters are deployed to deliver light signals through different outgoing ports while using the same wavelength. Typically, by employing iterative light-tree formation algorithms like Member-Only~\cite{xjzhang2000}, the set of light-trees computed for the same multicast session are not edge disjoint. Referring to the wavelength distinct constraint~\cite{bMukherjee2000}, the maximum number of wavelengths required per fiber by a multicast session equals to the number of light-trees constructed for this multicast session. Given a multicast group, the more destinations a light-tree is able to span, the fewer wavelengths will be needed. A lot of research has addressed the light-tree formation algorithms \cite{xjzhang2000,sgYan2003,xhJia2001,fZhou2008PNC,fZhou2009IC3N,fZhou2008ICCS}. In~\cite{xhJia2001}, a QoS routing algorithm is proposed to construct a light-tree with delay constraint and the cost approximation of light-tree is given. However, all the nodes are assumed to be capable of splitting this is why only one light-tree is required. In~\cite{xjzhang2000} the Reroute-to-Source algorithm proposes to make use of the shortest path tree from the source to all the destinations. Since the non-splitter MI nodes are not considered during the construction of the light-trees, the nodes connected to a MI node have to use different wavelengths to communicate with the source. Thereby, the number of destinations included in a light-tree is limited, and this is costly in terms of link stress (i.e., the maximum number of wavelengths required per fiber) and the total cost (i.e., the total wavelength channels used for a multicast session). In~\cite{sgYan2003} a Tabu search heuristic is addressed to find the optimal light-tree. As we know Tabu search is just a local search method, the result depends a lot on the initial solution, and hence the number of destinations included in one light-tree cannot be guaranteed. Besides, Tabu search is not time efficient for dynamic multicast traffic. The Member-Only algorithm proposed in~\cite{xjzhang2000} takes the splitting capability of nodes into account when constructing a light-tree. It is currently thought to have the best link stress and total cost for multicast routing in all-optical networks with spare light splitting and without wavelength conversion. However, Member-Only is just a simple application of the Minimum Path Heuristic~\cite{hTakahashi1980}, which uses exclusively the shortest path to connect the nearest destination to a light-tree. It may face failures (Figs.~\ref{fig: hypo1} and \ref{fig: hypo2}), since a shortest path is more likely to traverse Tac-exhausted MI nodes (i.e., a non-leaf MI nodes in a light-tree under construction), which is not permitted by the light splitting constraint. In fact, in case the shortest path does not work, a longer constraint satisfied path can be a good choice to connect a destination to the light-tree. Hence, the number of destinations in a Member-Only light-tree still could be improved (Fig.~\ref{fig: number of destinations}). In \cite{fZhou2008PNC}, a trade off is found between the link stress and the average delay. The DijkstraPro algorithm is proposed to construct the shortest path tree while taking the sparse light splitting into consideration. Although the delay is diminished, the links stress is still a little bigger than that of Member-Only algorithm. In \cite{fZhou2008ICCS}, a distance priority is introduced during the spanning of multicast light-trees, which improves the delay of Member-Only light-trees while keeping almost the same link stress.

In this article, we try to improve the link stress and the network throughout in all-optical WDM networks with multicast traffic and sparse light splitting constraint. It is assumed that the costly wavelength conversion devices are not available in our optical networks and the backward direction of an optical fiber cannot be used on the same wavelength as its forward. Since the more destinations a light-tree can hold the fewer wavelengths a multicast session will require, we introduce the Hypo-Steiner Heuristic to include as many destinations as possible in a light-tree for a multicast session. By deleting on-leaf MI nodes in the subtree from the original graph and computing the nearest destination plus the shortest paths in the modified graph, the proposed Hypo-Steiner Heuristic is able to easily find more available paths fulfilling the optical constraints to connect destinations to the same light-tree. It always computes the shortest one among the paths satisfying the constraints, which is the shortest path between two nodes in the modified graph after deleting. However, this path is not necessary to be the shortest path between these two nodes in the original graph. That is why we call it Hypo-Steiner Heuristic.

The rest of this article is organized as follows. In the next Section, multicast routing and wavelength assignment problem in sparse light splitting WDM network is formulated and the previous work is reviewed. Then Hypo-Steiner Heuristic is presented and analyzed in Section~\ref{sec:Hypo-Steiner Light-Tree Algorithm}. Numerical results are obtained and evaluated through the comparison among the proposed algorithm and the most efficient and state-of-art algorithms in Section~\ref{sec: Simulation And Numerical Results}.  Finally, a conclusion is made in section~\ref{sec:Conclusion}.

\section{All-Optical Multicast Routing With Sparse Light Splitting}
\label{sec:All-Optical Multicast Routing With Sparse Light Splitting}
\subsection{System Model and Problem Description}
\label{sec:System Model}
The sparse light splitting WDM network can be modeled by an undirected graph $G(V,E,c,W)$. $V$ represents the vertex-set of $G$. Each node $v \in V$ is deployed to be either an MI or an MC node beforehand. $E$ represents the edge-set of $G$, which corresponds to the fiber links between the nodes in the network.
  \begin{eqnarray}
    \nonumber V = \{v | v = MI~or~v = MC\}\\
        |V|=N, |E|=M
  \end{eqnarray}

$W$ denotes the number of wavelengths supported in each fiber link. Each edge $e \in E$ is associated with a cost functions $c(e)$. We define the cost $c(e)$ of each edge as
  \begin{equation}
    c(e) = 1~unit~hop-count~fiber~cost
  \end{equation}

  Function $c$ is additive over the links of a lightpath $LP(u,v)$ between two end point nodes $u$ and $v$,
  \begin{equation}
   c\big(LP(u,v) \big) = \sum_{e \in LP(u,v)}{c(e)}
  \end{equation}

  We consider a multicast session $ms(s,D)$, which requests for setting up a set of light-trees under optical constraints (continuous wavelength constraint, distinct wavelength constraint, sparse light splitting and lack of wavelength conversion) from the source $s$ to a group of destinations $D$ simultaneously. Without loss of generality, assume there are $|D|=K$ destinations in a multicast session and an iterative light-tree construction heuristic builds $k$ light-trees $LT_{i}(s,D_{i})$ in sequence for $ms(s,D)$, where $i \in [1,k]$, and $1 \leq k \leq K \leq N-1$. Although the $i^{th}$ light-tree $LT_{i}(s,D_{i})$ may span some destinations already spanned in the previous light-trees,  $D_{i}$ is used to denote exclusively the set of newly served destinations in the $i^{th}$ light-tree (these newly served destinations can not be spanned in the previous light-trees, the $i^{th}$ light-tree is built to serve them uniquely). Hence, these $k$ sets $D_{i}$ are disjoint,
  \begin{equation}
    \label{equation: disjoint}
    \forall i, j \in [1,k]~and~i \neq j,~D_{i} \cap D_{j} = {\O}
  \end{equation}

  Since all the destinations in $D$ are spanned by $k$ light-trees, we obtain
  \begin{equation}
    D = \bigcup_{i=1}^{k}{D_{i}}
  \end{equation}

  Let a positive integer $l_{i}$ denote the size of the subset $D_{i}$, then we have
 \begin{eqnarray}
 \label{equation: li}
 \nonumber \forall i \in [1,k], |D_{i}|=l_{i} \geq 1\\
           \sum_{i=1}^{k}{l_{i}} = |D| = K
 \end{eqnarray}

First, regarding the optimization of network resources, the total cost (i.e., the total wavelength channels used for a multicast session) and the link stress (i.e., the maximum number of wavelengths required per link by a multicast session) should be minimized. The total cost of a multicast session can be calculated by the sum of cost in all the light-trees build for the multicast session.
\begin{eqnarray}
 \label{equation: cost}
 \nonumber c\big(ms(s,D)\big) &=& \sum_{i=1}^{k}{c\big[LT_{i}(s,D_{i})\big]}\\
                    &=& \sum_{i=1}^{k}\sum_{e\in LT_{i}(s,D_{i})}{c(e)}\\
 \nonumber          &=& \sum_{i=1}^{k}\sum_{e\in LT_{i}(s,D_{i})}{1}
 \end{eqnarray}

Since the light-trees for the same multicast session are not edge-disjoint, they could not be assigned the same wavelength. For a multicast session, link stress equals to the number of light-trees $k$ built for the multicast session.
    \begin{equation}
        Stress\big(ms(s,D)\big) = k
    \end{equation}

Secondly, considering the network throughput, given $W$ as the number of wavelengths supported in one fiber, the more multicast sessions can be accepted, the better.

  \begin{equation}
    Throughput = Num(accepted~sessions)|_{W}
  \end{equation}

 However, not all these parameters could be optimized simultaneously. The main objective in this paper is trying to build a light-tree, which is able to include as many destinations as possible for a multicast session. Accordingly, the number of light-trees built for a multicast session will be minimized.

 \begin{equation}
  \max\{|D_{i}|=l_{i}\}~and~\min\{k\}
 \end{equation}

 \subsection{Member-Only Algorithm}
 \label{sec: Previous Work}

The heuristic considered as the best one is the Member-Only algorithm which derives from the Minimum Path Heuristic~\cite{hTakahashi1980}. It involves three nodes sets during the construction of the light-tree. In a subtree $LT$ under construction,

   $MC\_SET$: includes source node, MC nodes and the leaf MI nodes. They may be used to span light-tree $LT$ and, thus are also called connector nodes in $LT$.

   $MI\_SET$: includes only the non-leaf MI nodes, whose splitting capability is exhausted. Hence, these nodes are not able to connect a new destination to the subtree $LT$.

   $D$: includes unserved multicast members which are neither joined to the current light-tree $LT$ nor to the previously constructed multicast light-trees.\\

In Member-Only algorithm, the shortest paths between all pairs of nodes in the network are computed beforehand and stored. The span of the distribution light-tree $LT$ begins with the source: $LT=\{s\}$, $MI\_SET=\O$, $MC\_SET=\{s\}$ and $D=\{all~destinations\}$. At each step, try to find the nearest destination from $d \in D$ to $c \in MC\_SET$, whose shortest path $SP(d, c)$ does not involve any node in $MI\_SET$. If it is found, $SP(d, c)$ is added to $LT$ and the sets are updated: $d$ and the MC nodes are added to $MC\_SET$,  non-leaf MI nodes are added to $MI\_SET$, and $d$ is removed from $D$. Otherwise (i.e., no such constraints-satisfying shortest path could be found), the current light-tree $LT$ is finished, and another light-tree assigned a new available wavelength is started using the same procedure until no destination is left in $D$.

\section{Hypo-Steiner Light-Tree Algorithm}
\label{sec:Hypo-Steiner Light-Tree Algorithm}
\subsection{Hypo-Steiner Heuristic}
\label{sec:Hypo-Steiner Heuristic}
\begin{algorithm}
    \algsetup{indent=2em}
    \caption{Hypo-Steiner Light-tree Algorithm}
    \label{alg1}
    \begin{algorithmic}[1]
    \REQUIRE {A graph $G(V,E,c,W)$ and a multicast session $ms(s,D_{0})$.}
    \ENSURE {A set of Light-trees $LT_{k}(s,D_{k})$ each on a different wavelength $w_{k}$ for $ms(s,D_{0})$.}
    \STATE {$k \leftarrow 1$} \COMMENT {$k$ is the serial number of a light-tree}
    \STATE {$D \leftarrow D_{0}$}
    \WHILE {$(D\neq \O)$}
        \STATE {$i \leftarrow 1$} \COMMENT {$i$ is the serial number of a renewed graph}
        \STATE {$G_{i} \leftarrow G(V,E,c,W)$}
        \STATE {$MC\_SET \leftarrow \{s\}$}
        \STATE {$LT_{k} \leftarrow \{s\}$}
        \WHILE {$(SP_{G_{i}}(\textbf{d},\textbf{c}) \neq \infty)$}
            \FORALL {$(d \in D$ and $c \in MC\_SET$)}
                \STATE {Compute the shortest path $SP_{G_{i}}(d,c)$ in $G_{i}$}.
            \ENDFOR
            \STATE {Find the nearest destination $\textbf{d}$ with $SP_{G_{i}}(\textbf{d},\textbf{c})$}
            \STATE {$LT_{k} \leftarrow LT_{k} \cup SP_{G_{i}}(\textbf{d},\textbf{c})$}
            \STATE {$MC\_SET \leftarrow MC\_SET \cup \{$MC in $ SP_{G_{i}}(\textbf{d},\textbf{c})\} \cup \{\textbf{d}\}$}
            \IF {($\textbf{c}$ is an MI node)}
                 \STATE {$MC\_SET \leftarrow MC\_SET \setminus \{\textbf{c}\}$}
            \ENDIF
            \STATE {$D \leftarrow D \setminus \{d\}$}
            \STATE {$G_{i+1} \leftarrow G_{i} \setminus \{$edges and MI nodes in $SP_{G_{i}}(\textbf{d},\textbf{c})\}$}
            \STATE {$i \leftarrow i+1$}
        \ENDWHILE
        \STATE {Assign wavelength $w_{k}$ to $LT_{k}$}
        \STATE {$k \leftarrow k+1$} \COMMENT {Star a new light-tree $LT_{k+1}$}
    \ENDWHILE
    \end{algorithmic}
    \end{algorithm}
During the construction of a light-tree $LT$, non-leaf MI nodes in $LT$ have exhausted their TaC capability and could not be used again to connect another destination to the subtree. In other words, these non-leaf MI nodes are useless for the spanning of the current light-tree $LT$. Thus, why do not we delete them from the graph? At each step, we compute the shortest paths and the distances from the destinations in set $D$ to the subtree $LT$ in a new graph, say $G_{i}$ (generated by deleting all non-leaf MI nodes in $LT$ from the original graph $G$). Then, find the nearest destination and join it to the subtree with the shortest path in $G_{i}$. Here, we can see, it is definitely true that the shortest paths between any two nodes in the new graph $G_{i}$ will not traverse any MI nodes in $MI\_SET$. Hence, by computing the shortest path in the modified graph $G_{i}$, when searching the nearest destination to the subtree $LT$, we need not to check whether its shortest path to $LT$ (precisely speaking, to its connector node in $MC\_SET$) satisfies the splitting constraints or not. The benefits of the Hypo-Steiner Heuristic are based on the following two observations.

First, Hypo-Steiner Heuristic can enumerate all the possible shortest paths which are able to connect a destination to the light-tree while respecting constraints. When using Dijkstra to compute the shortest path from node $d_{i} \in D$ to a connector node $c_{j} \in MC\_SET$, only one could be found, although there may be several ones. If the useless non-leaf MI nodes are still in the graph, then the shortest path that traverses a node in $MI\_SET$ cannot be avoided in the computation. If, unfortunately, the founded shortest path from a destination $d_{i}$ to the current subtree $LT$ traverses a node in $MI\_SET$, then $d_{i}$ could not be joined to the current light-tree $LT$ according to the Member-Only algorithm~\cite{xjzhang2000} and consequently a new light-tree on another wavelength will be needed. However, after deleting all the useless non-leaf MI nodes and the used edges from the original graph $G$, the other shortest path, which does not involve any node in $MI\_SET$, could be definitely found in the new graph if it exists.

    \begin{figure}
        \begin{center}
        \includegraphics[width=.4\textwidth]{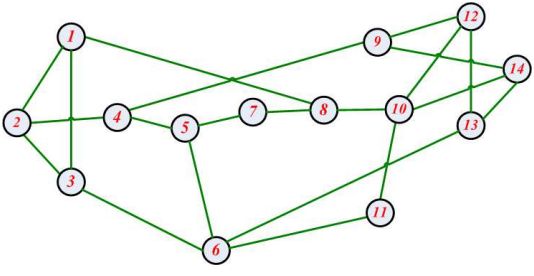}
        \end{center}
        \caption{NSF Network Topology}
        \label{fig: NSFnet}       
    \end{figure}

\textbf{Example 1:}
In the NSF network of Fig.~\ref{fig: NSFnet}, a multicast session $ms_{1}\big(s: 8, D:(4, 6)\big)$ request arrives, and only source $s$ is an MC node. Using Member-Only algorithm, node 4 is firstly connected to node 8 using the shortest path $SP(8-7-5-4)$. Now, $MC\_SET = \{4, 8\}$, $MI\_SET = \{5, 7\}$. Next, compute the shortest paths from node 6 to the nodes in $MC\_SET$. $SP(4-5-6)$ involves non-leaf MI node 5, so it does not work. If $SP(8-7-5-6)$ is unfortunately computed out, then a new tree should be employed to accommodate node 6 as in Fig.~\ref{fig: hypo1}(a). But, with the help of Hypo-Steiner Heuristic, node 5 and 7 are deleted from the original graph $G_{1}$. Hence, $SP(8-7-5-6)$ (Fig.~\ref{fig: hypo1}(a)) could be definitely avoided, and another shortest path $SP(8-10-11-6)$ (Fig.~\ref{fig: hypo1}(b)) could be found in the new graph $G_{2}$, which does not traverse any node in $MI\_SET$.

Secondly, in case that the shortest path does not work, the Hypo-Steiner Heuristic tries to find a second shortest constraint-satisfied path to connect a destination to the light-tree while Member-Only will stop the spanning of the current light-tree. Only making use of the shortest path in the original graph, not all the possible paths satisfying the splitting constraints could be enumerated by Member-Only algorithm. This is because that the shortest path from node $d_{i} \in D$ to a connector node $c_{j} \in MC\_SET$ in the original graph is always used to span the subtree. However, most of these shortest paths involve some nodes in $MI\_SET$, which are not permitted by the splitting constraint. In fact, if all the MI nodes and the used edges in $LT$ are deleted from the graph to get $G_{2}$, then in the following step, a secondly shortest path between two nodes (it is the shortest path in the new graph $G_{2}$) may be found by just using Dijkstra's algorithm. The most important thing is that it satisfies the light splitting constraint.

   \begin{figure}
        \begin{center}
        \includegraphics[width=.4\textwidth]{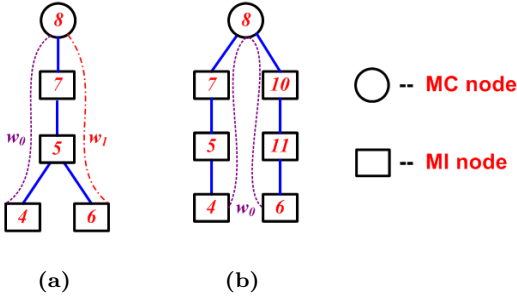}\\
        \end{center}
        \hspace{.5in}\mbox{\textbf{(a)}} \hspace{.76in}\mbox{\textbf{(b)}}
        \caption{For multicast session $ms_{1}$, (a) Light-tree built by Member-Only; (b) Light-tree built using Hypo-Steiner Heuristic.}
        \label{fig: hypo1}
    \end{figure}

\textbf{Example 2:}
Assume a multicast session $ms_{2}\big(s:8,D:(9\sim11)\big)$ is required in the NSF network, where only node 8 is an MC node. According to the Member-Only algorithm, node 10 is first added to the subtree, then node 11. At this moment, $MC\_SET=\{8, 11\}$, $MI\_SET=\{10\}$.
Noting that shortest paths $SP(8-10-12-9)$, $SP(8-10-14-9)$ and $SP(11-10-9)$ involve non-leaf MI node 10, which belongs to $MI\_SET$. Consequently, node 9 cannot be joined to the current subtree (using wavelength $w_{0}$), and a new light-tree on wavelength $w_{1}$ should be used as shown in Fig.~\ref{fig: hypo2}(a). However, by implementing the Hypo-Steiner Heuristic, MI node 10, links $(8,10)$ and $(10,11)$ are deleted from graph $G_{1}$ to get $G_{2}$. It is interesting to find that the shortest path $SP_{G_{2}}(8-7-5-4-9)$ in the new graph $G_{2}$ is able to connect node 9 to the current subtree, which is the shortest constraint-satisfied path between them. So, only one wavelength $w_{0}$ is required as shown in Fig.~\ref{fig: hypo2}(b).

    \begin{figure}
        \begin{center}
        \includegraphics[width=.25\textwidth]{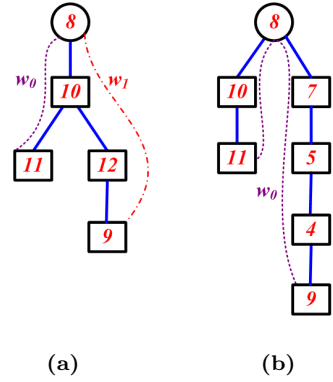}\\
        \end{center}
        \hspace{1in}\mbox{\textbf{(a)}} \hspace{.9in}\mbox{\textbf{(b)}}
        \caption{For multicast session $ms_{2}$, (a) Light-tree built by Member-Only; (b) Light-tree built using Hypo-Steiner Heuristic.}
         \label{fig: hypo2}
    \end{figure}

\subsection{Overview of Hypo-Steiner Light-tree Algorithm}
\label{sec: Overview of Hypo-Steiner Light-tree Algorithm}
\begin{theorem}
\label{theorem: 1}
Given $MI\_SET$ and $MC\_SET$ for a subtree $LT$ under construction, if $\exists$ one or several paths $P(d, c)$ satisfying (a) $d \in D$ and $c \in MC\_SET$, (b) $\forall v \in  P(d, c)$, $v \not \in MI\_SET$, HSLT algorithm is able to span at least one more destination to $LT$, and the path used to connect a destination to $LT$ is the shortest one among all the existing paths $P(d, c)$.
\end{theorem}

\begin{proof}
In order to get a new modified graph $G_{i}$, algorithm HSLT deletes all the nodes in $MI\_SET$ from the previous graph $G_{i-1}$, thus no path in $G_{i}$ contains any non-leaf MI node in $LT$. If condition (b) $\forall v \in P(d, c), v \not \in  MI\_SET$ is satisfied, then no node in $P(d, c)$ will be deleted from the previous graph $G_{i-1}$. Thus the paths $P(d, c)$ will appear in the modified graph $G_{i}$. HSLT computes all the shortest paths from $d \in D$ and $c \in MC\_SET$ in the modified graph $G_{i}$. It is evident that the shortest one among the paths $P(d, c)$ will be found. Then the nearest destination is connected to $LT$ with the shortest path found.  After that, HSLT will continue the span procedure of $LT$ until blocked.
\end{proof}

\begin{theorem}
\label{theorem: 2}
Member-Only will fail to span another destination in $LT$, if none of the paths $P(d, c)$ is the shortest path in the network.
\end{theorem}

\begin{proof}
If all the paths $P(d, c)$ satisfying condition (a) and (b) are not the shortest path in the original topology, it means all the shortest paths $SP(d, c)$ with $d \in D$ and $c \in MC\_SET$ in the original graph have the following problem
\begin{equation}
  SP(d, c) \cap MI\_SET  \neq \O
\end{equation}
As a result, Member-Only will block the construction of the current light-tree $LT$ and continue another one to serve the destinations left in $D$ (for instance Examples 1 and 2).
\end{proof}

Through the above two theorems, we can see, under the same situation, HSLT has a bigger capability to span more destinations in one light-tree than Member-Only. This is why a HSLT light-tree is able to serve more destinations in most cases, especially for the first built light-tree. Once the number of destinations served in one light-tree increases, the number of light-trees required for a given multicast session will decline, which leads to the decrease of link stress. Since one multicast session requires fewer wavelengths, for a given number of wavelengths $W$ supported per fiber link, more multicast sessions are reasonable likely to be established simultaneously. This may help to improve the network throughput in the scenario that the network traffic load is not too high.

Although many heuristics~\cite{xjzhang2000,aHamad2006,fZhou2008PNC,nSreenath2001Photonic,fZhou2008ICCS} have been proposed for the construction of light-trees under the sparse light splitting constraint, none of them discuss about the cost bound of the light-trees. Since it is NP-Hard to compute the light-trees with the minimum cost, the cost bound is an important quality of the light-tress. Next, we give the cost bound of the HSLT light-tress.

\begin{theorem}
\label{theorem: 3}
Given a multicast session $ms(s,D)$, the total cost of the light-trees built by HSLT algorithm is bounded to
\begin{equation}
 K = |D| \leq c\big(ms(s,D)\big) \leq \frac{N(N-1)}{2}.
\end{equation}
\end{theorem}

\begin{proof}
When finishing the first light-tree $LT_{1}$, the destinations left in $D \setminus D_{1}$ could not be included in $LT_{1}$, otherwise we don't have to span them in a second light-tree. For the same reason, after the $j^{th}$ light-tree is built, the destinations left in $D \setminus \big(\bigcup_{i=1}^{j}{D_{i}}\big)$ could not be included in the $j^{th}$ light-tree. Since equation~(\ref{equation: disjoint}), in the $j^{th}$ light-tree, there are at most $N - \sum_{i=j+1}^{k}{l_{i}}$ nodes. Thus the cost of the $j^{th}$ light-tree complies
\begin{equation}
c\big(LT_{j}(s,D_{j})\big) \leq N - 1 - \sum_{i=j+1}^{k}{l_{i}}
\end{equation}
Then, we obtain the total cost of the set of light-trees built for multicast session $ms(s,D)$
\begin{eqnarray}
\nonumber c\big(ms(s,D)\big) &=& \sum_{i=1}^{k}{c\big[LT_{i}(s,D_{i})\big]}\\
                             &\leq& \sum_{i=1}^{k}{\big(N - 1 - \sum_{i=j+1}^{k}{l_{i}}\big)}\\
\nonumber                    &=& k(N - 1) - \sum_{j=1}^{k}\sum_{i=j+1}^{k}{l_{i}}
\end{eqnarray}
According to equation~(\ref{equation: li}), it holds that
\begin{equation}
\sum_{i=1}^{k}{\sum_{i=j+1}^{k}{l_{i}}} \geq \frac{1}{2}k(k-1)
\end{equation}
Thus,
\begin{eqnarray}
\nonumber c\big(ms(s,D)\big) &\leq& k(N-1)-\frac{1}{2}k(k-1)\\
                             &=&-\frac{1}{2}\bigg[k-\big(N-\frac{1}{2}\big)\bigg]^{2} + \frac{1}{2}\big(N-\frac{1}{2}\big)^{2}
\end{eqnarray}
When $k=N-1$, we get the following inequality
\begin{equation}
c\big(ms(s,D)\big) \leq \frac{N(N-1)}{2}.
\end{equation}
Moreover, in the $j^{th}$ light-tree $LT_{j}(s,D_{j})$, there are at least $l_{j}+1$ nodes when the $j^{th}$ light-tree includes no other nodes than the source node and the destinations in set $D_{j}$, i.e., $LT_{j}(s,D_{j}) = \{s\} \cup {D_{i}}$. Then, the cost of $LT_{j}(s,D_{j})$ holds
\begin{equation}
c\big(LT_{j}(s,D_{j})\big) \geq l_{j}
\end{equation}
Thereby, we obtain the total cost of multicast session $ms(s,D)$,
\begin{eqnarray}
c\big(ms(s,D)\big) \geq \sum_{j=1}^{k}{l_{i}} = K = |D|
\end{eqnarray}
\end{proof}

\subsection{Time Complexity of Hypo-Steiner Light-tree Algorithm}
\label{sec: Time Complexity of Hypo-Steiner Light-tree Algorithm}
At every step, a new graph $G_{i}$ is generated. In order to find the nearest destination to the subtree built, Dijkstra's algorithm should be used compute all the shortest paths in the new graph $G_{i}$ from each destination $d\in D$ to the connector nodes in $MC\_SET$. So, the time complexity of each step is $|D|_{i} \times O(Dijkstra)$, where $|D|_{i}$ denotes the number of destinations left in set $D$ at the $i^{th}$ step. There are $K=|D|$ steps in total, and $|D|_{i} = K - i$. The most recent worst-case time complexity of Dijkstra is $O(NlogN+M)$~\cite{mBarbehenn1998}. Therefore, the overall time complexity is
\begin{equation}
\frac{1}{2}K(K+1) \times O(Dijkstra) = O\big(K^{2}(NlogN+M)\big)
\end{equation}
In fact, if we modify the Dijkstra's algorithm, the complexity could be reduced a lot. At each iteration, in order to find the nearest destination in set $D$, the connector nodes in $MC\_SET$ could be viewed ensemble as the source (set their distances to the source as 0), and thus only one shortest path tree is required in order to compute the distances and search the nearest destination. There are $K=|D|$ destinations in one multicast session, so Dijkstra's algorithm is only used for $K$ times. Thus, the time complexity of the proposed algorithm could be diminished to
\begin{equation}
K \times O(Dijkstra) = O\big(K(NlogN+M)\big)
\end{equation}

\section{Simulation And Numerical Results}
\label{sec: Simulation And Numerical Results}
In this section, we use the simulation to evaluate the average case performance of the proposed Hypo-Steiner light-tree (HSLT) algorithm. It is compared with the existing algorithms like Reroute-to-Source (R2S) and Member-Only (MO) algorithm~\cite{xjzhang2000}. The latter one is currently shown to provide the best link stress and total cost.

\subsection{Simulation Model and Evaluation Metrics}
\label{sec: Simulation Model and Evaluation Metrics}
The USA Longhaul network (28 nodes and 43 links) in Fig.~\ref{longhual} is employed as platform for the simulation. Both the members of a multicast group and the MC nodes are selected uniformly and independently in the topology. For each network configuration (given a number of MC nodes and a group size), 10000 random multicast sessions are generated and the result is the average of 10000 computations. In the simulation, five performance metrics are considered:
\begin{itemize}
\item \textbf{Link Stress}. It is the number of wavelengths required per fiber, which equals to the number of light-trees $k$ built for a multicast session.
\item \textbf{The number of destinations included in the first light-tree}, with $|D_{1}|=l_{1}$. In most cases, with HSLT and MO two light-trees are sufficient to cover all the group members (Fig.~\ref{fig: link stress}). The destinations served by the second light-tree are always fewer.  Thus only the first light-tree can signify the capacity of spanning destinations.
\item \textbf{Total Cost}. It is the wavelength channels used in all the light-trees built for a multicast session $ms(s,D)$.
\item \textbf{Throughput}. It is the number of multicast sessions accepted simultaneously for a given $W$.
\item \textbf{Efficiency of Wavelength Usage}. It can be calculated as
\begin{equation}
Eff(\lambda ) = \frac{1}{MW}\sum_{e\in E}\sum_{\lambda=1}^{W}{e(\lambda)}
\end{equation}
, where $e(\lambda)$ is defined as the utilization state of wavelength $\lambda$ on fiber link $e$
\begin{eqnarray}
\nonumber {\forall e\in E, ~and~ \lambda \in [1,W]},\\
 e(\lambda) = \left \{
\begin{array}{ll}
 1, $ if $ \lambda $ is used on link $ e\\
 0, $ if $ \lambda $ is available on link $ e
 \end{array}
 \right.
\end{eqnarray}

And we set $W$ = 20 in our simulation.
\end{itemize}

   \begin{figure}
            \begin{center}
            \includegraphics[width=.4\textwidth]{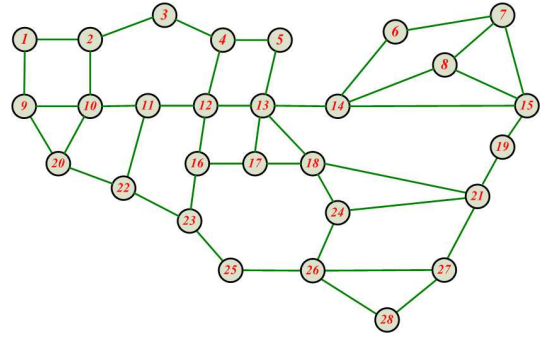}
            \end{center}
            \caption{USA Longhaul Network}
            \label{longhual}       
    \end{figure}

     \begin{figure*} [!t]
    \begin{center}
    $\begin{array}{c@{\hspace{1in}}c}
    \epsfxsize=2.17in \epsffile{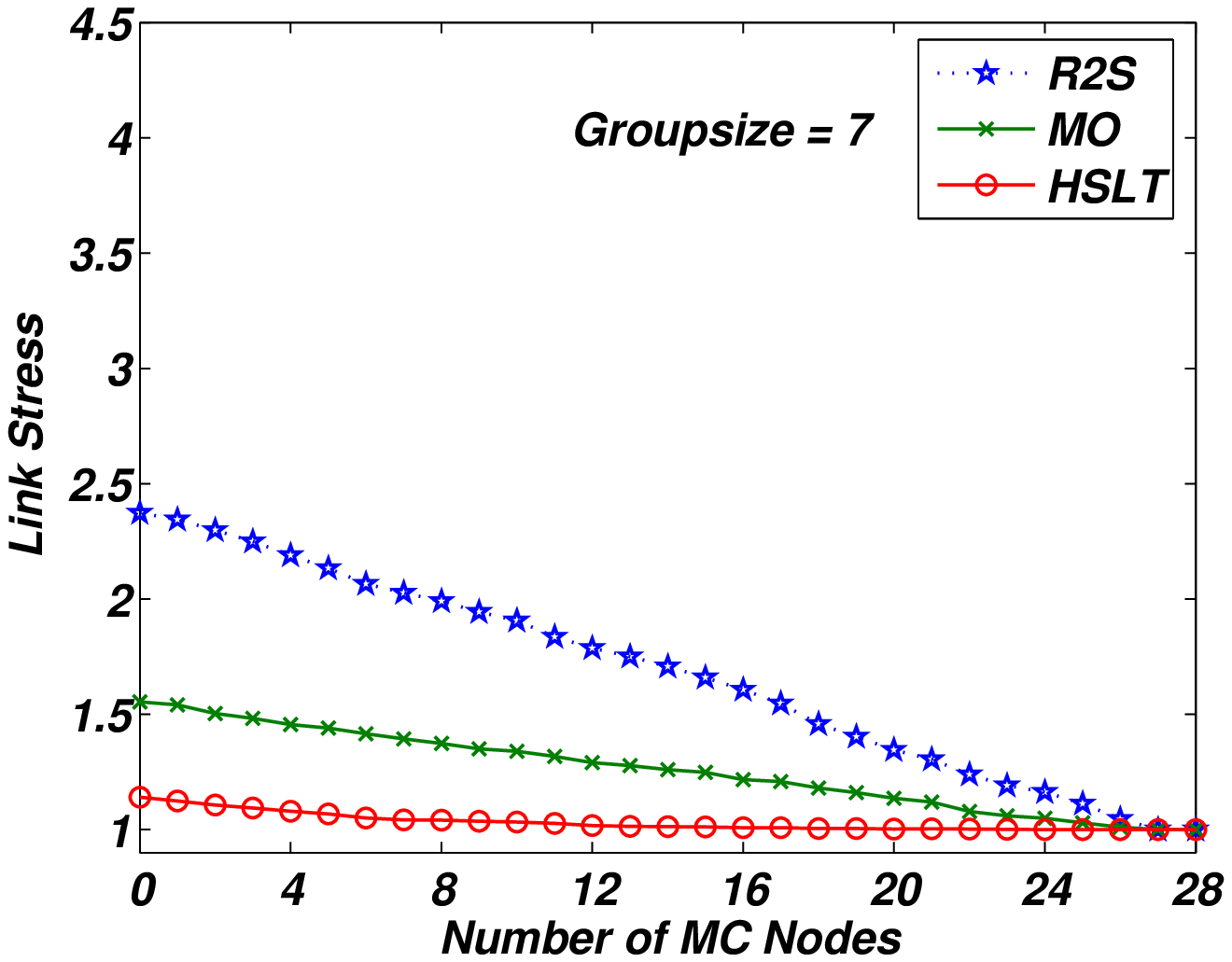} & \epsfxsize=2.17in \epsffile{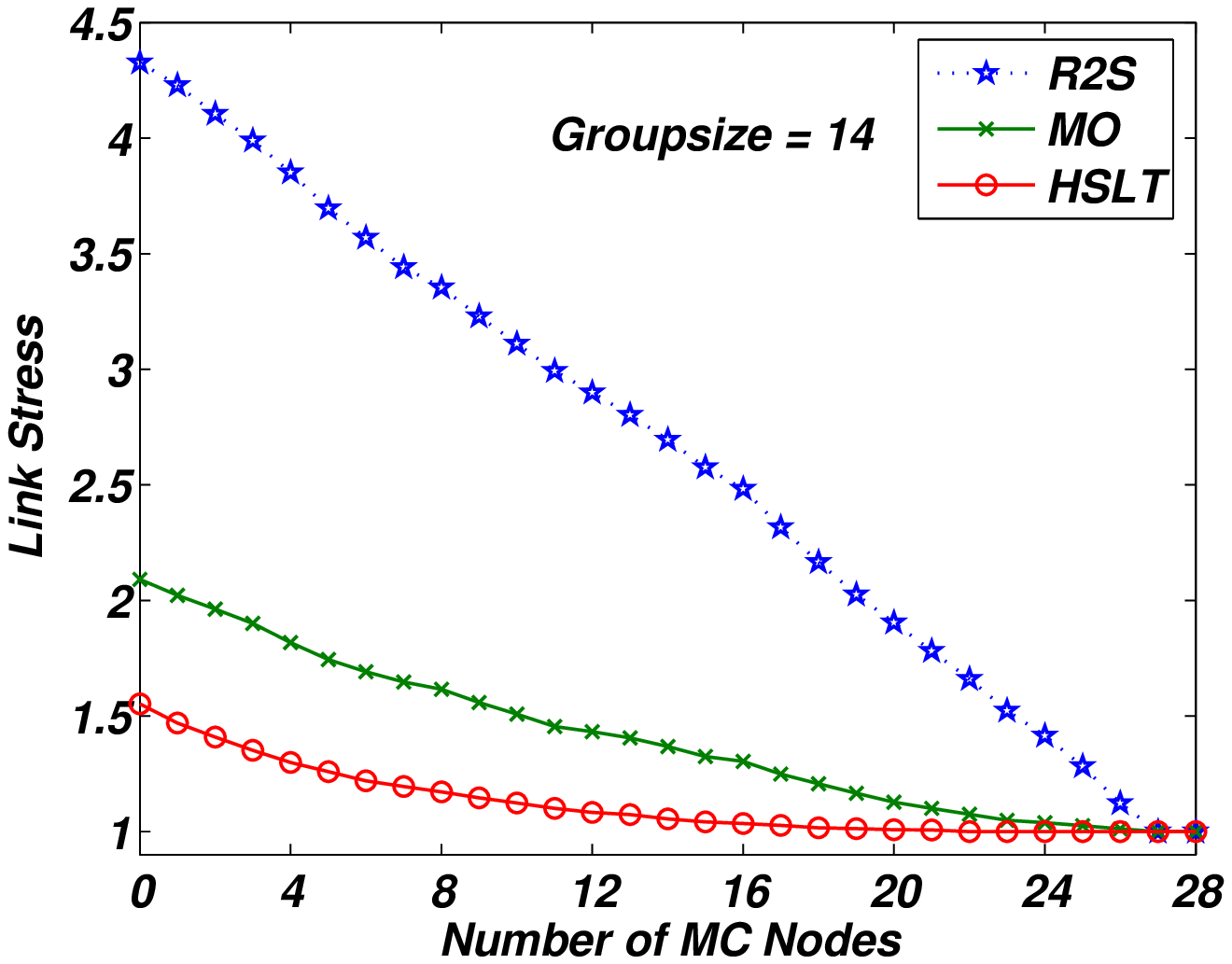} \\
    \mbox{\bf (a)} & \mbox{\bf (b)}
    \end{array}$
    \end{center}
    \caption{Comparison of the Link Stress when (a) groupsize = 7; (b) groupsize = 14 in the USA Longhaul Network}
    \label{fig: link stress}
    \end{figure*}

    \begin{figure*}
    \begin{center}
    $\begin{array}{c@{\hspace{1in}}c}
    \epsfxsize=2.17in \epsffile{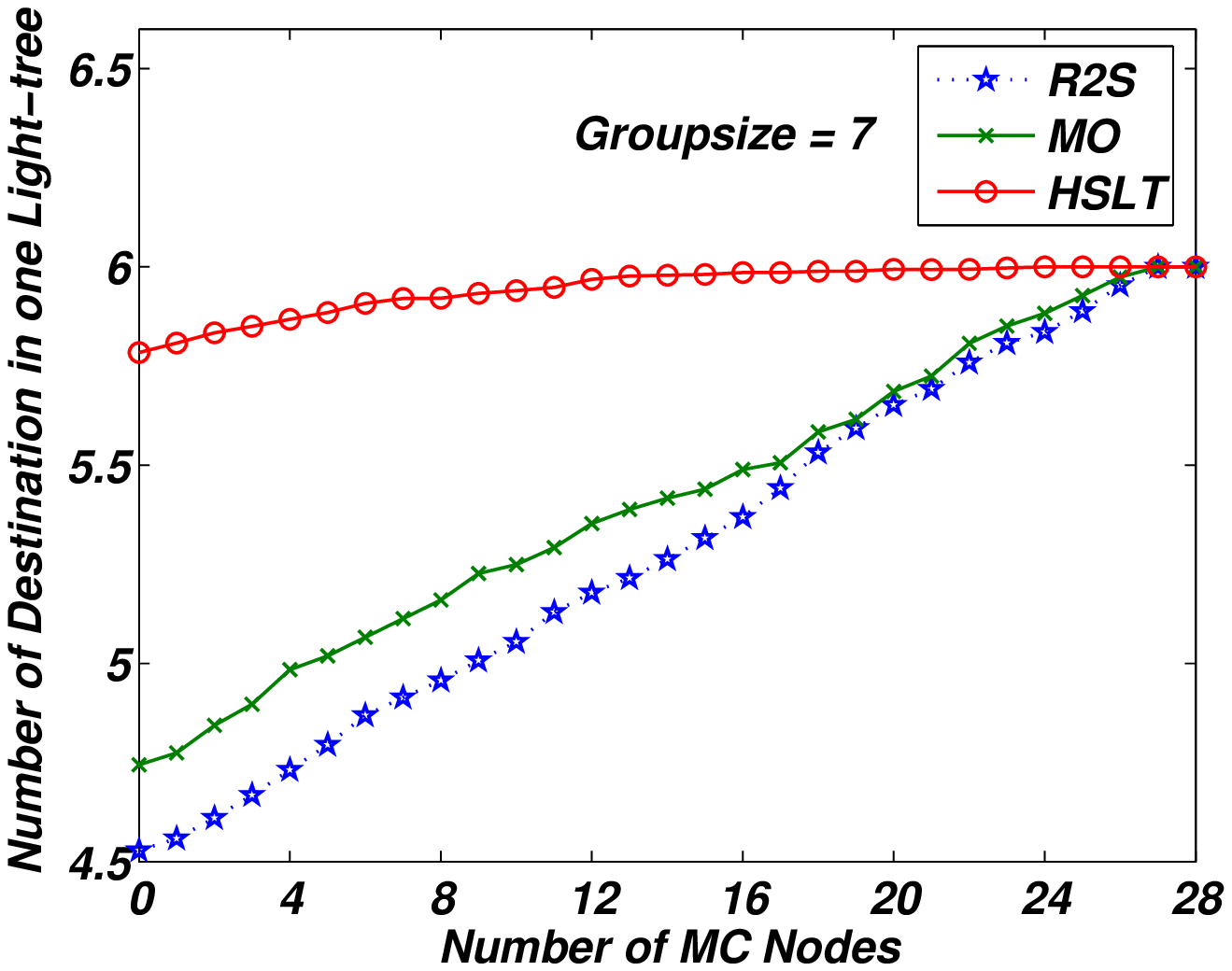} & \epsfxsize=2.17in \epsffile{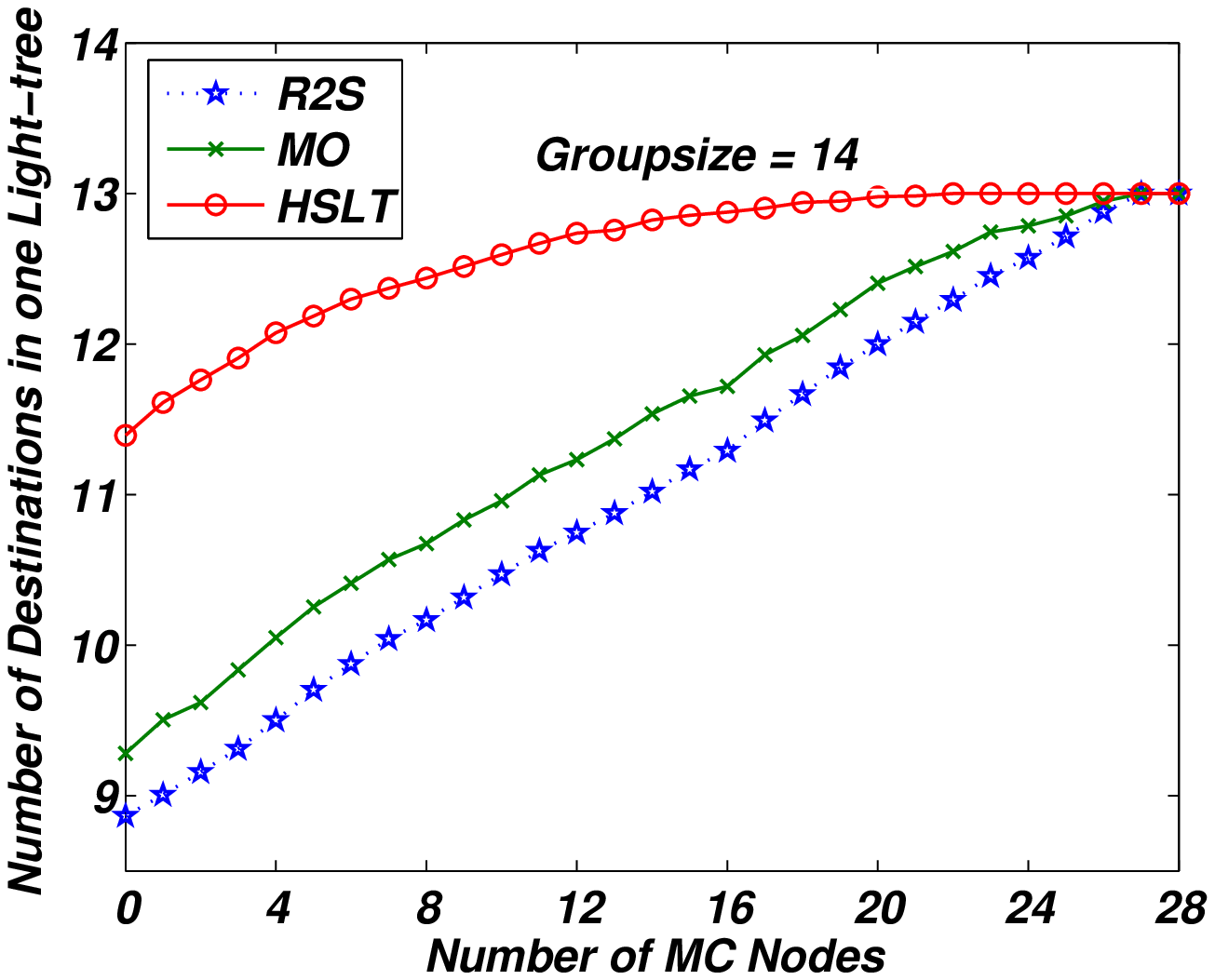} \\
    \mbox{\bf (a)} & \mbox{\bf (b)}
    \end{array}$
    \end{center}
    \caption{Comparison of the Number of Destinations included in a light-tree when (a) groupsize = 7; (b) groupsize = 14 in the USA Longhaul Network}
    \label{fig: number of destinations}
    \end{figure*}

    \begin{figure*}[!t]
    \begin{center}
    $\begin{array}{c@{\hspace{1in}}c}
    \epsfxsize=2.17in \epsffile{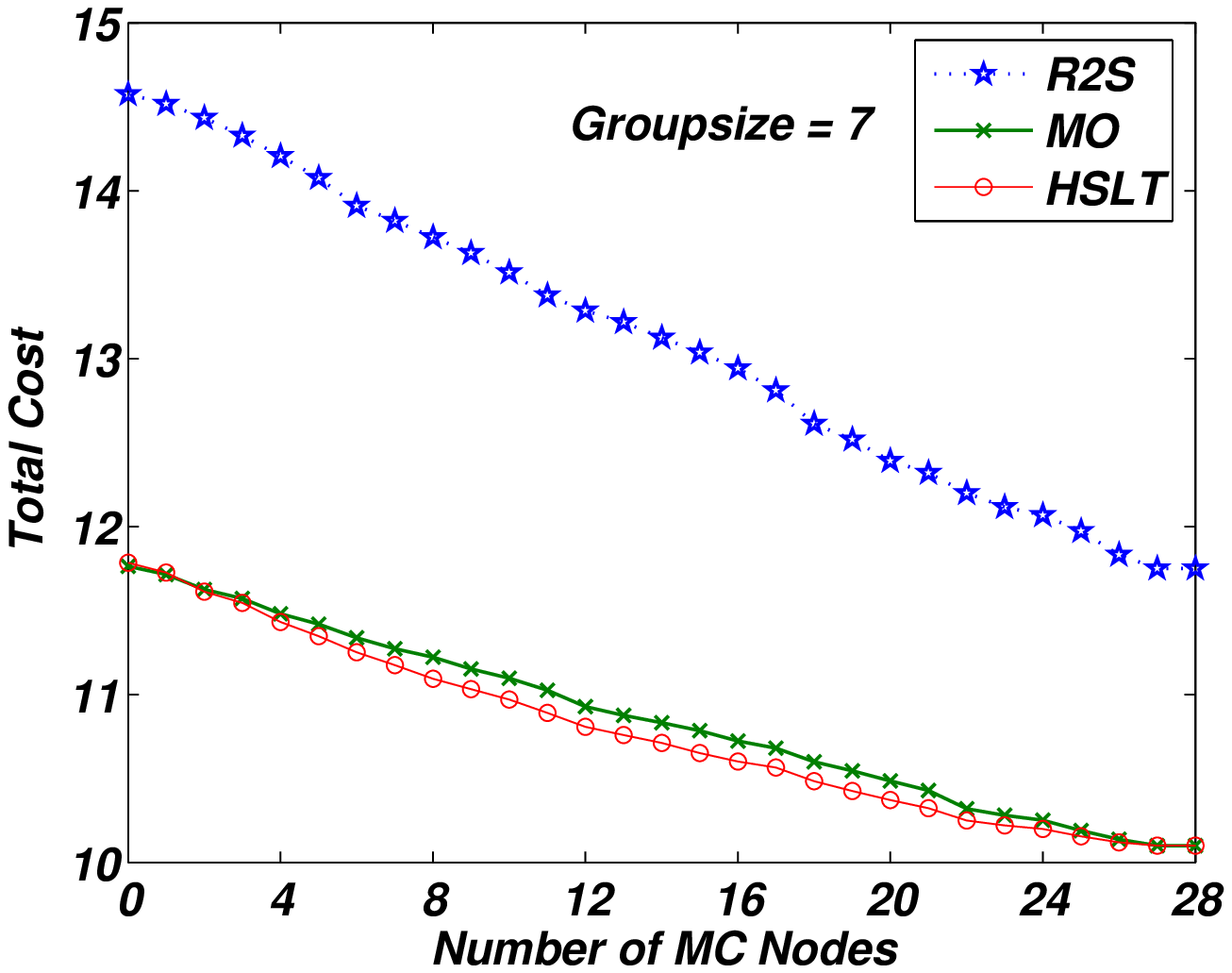} & \epsfxsize=2.17in \epsffile{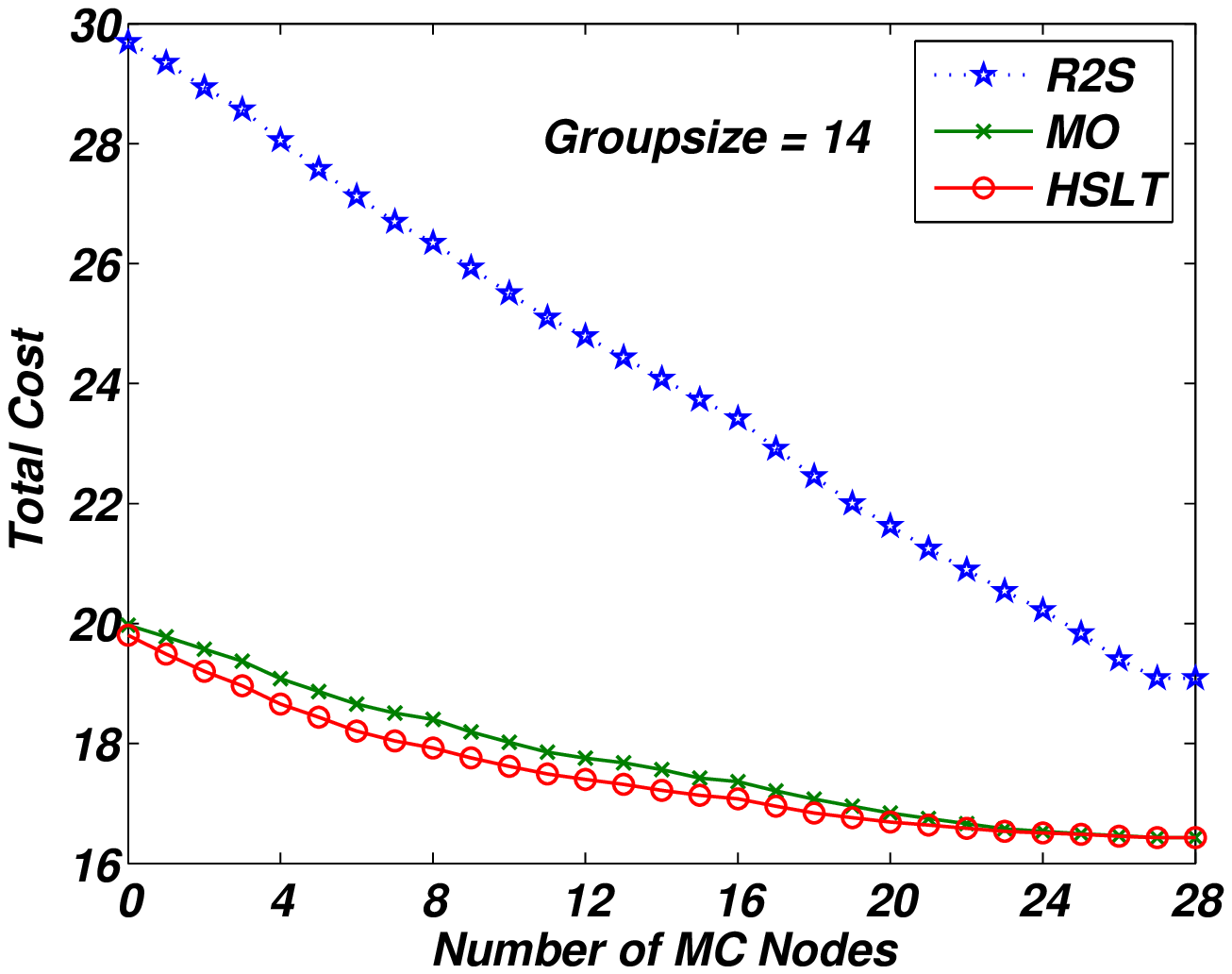} \\
    \mbox{\bf (a)} & \mbox{\bf (b)}
    \end{array}$
    \end{center}
    \caption{Comparison of the Total Cost when (a) groupsize = 7; (b) groupsize = 14 in the USA Longhaul Network}
    \label{fig: total cost}
    \end{figure*}

    \begin{figure*}
    \begin{center}
    $\begin{array}{c@{\hspace{1in}}c}
    \epsfxsize=2.17in \epsffile{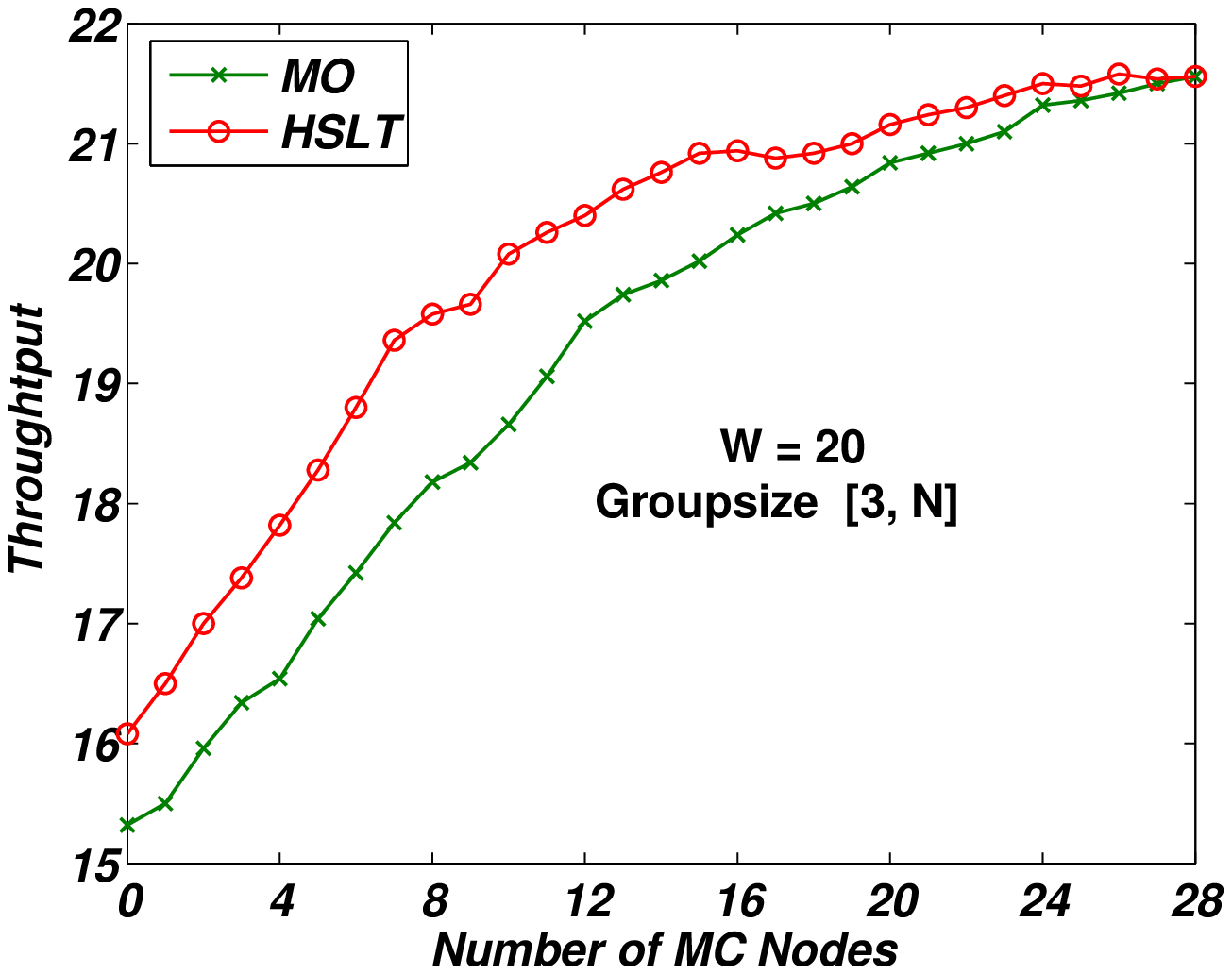} & \epsfxsize=2.17in \epsffile{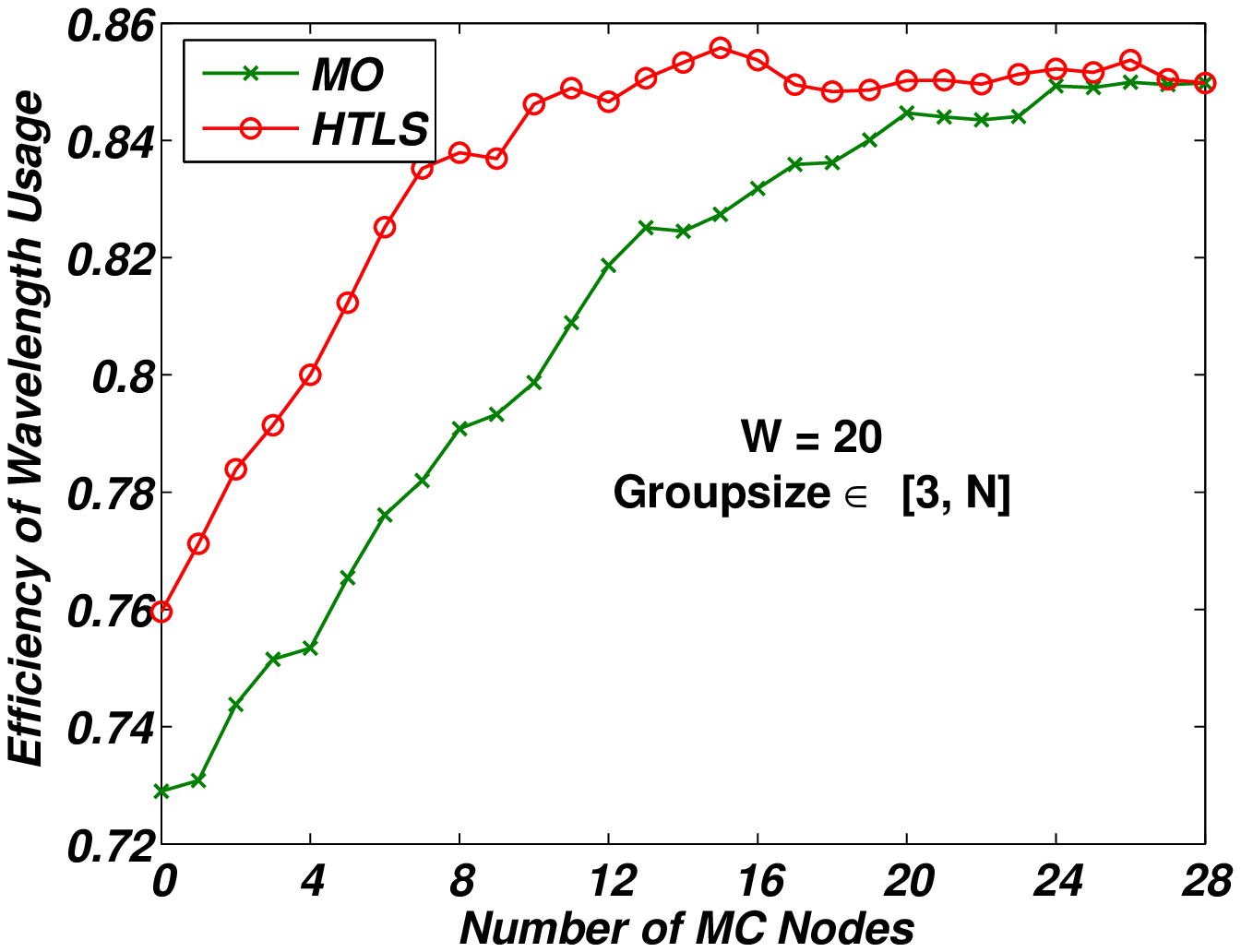} \\
    \mbox{\bf (a)} & \mbox{\bf (b)}
    \end{array}$
    \end{center}
    \caption{Comparison of the (a) Throughput; (b) Efficiency of Wavelength Usage in the USA Longhaul Network}
    \label{fig: Throughput and Efficiency}
    \end{figure*}

\subsection{Performance Analysis}
\label{sec: Performance Analysis}
The multicast session group size (including the source), $K+1$, is set to 7 and 14 respectively, which corresponds to the group member ratio of 25\% and 50\%, where
\begin{equation}
Ratio=\frac{(K+1)}{N}.
\end{equation}

First, we compare the quality of different light-trees built by Reroute-to-Source (R2S), Member-Only (MO) and HSLT, respectively.

(1) The link stress of the three types of light-trees is plotted in Fig.~\ref{fig: link stress}. HSLT can always get the best link stress among the three algorithms. If there is no splitter, the HSLT algorithm is able to save 0.6 and 2.9 wavelengths for a multicast session instead of using the other two algorithms when Ratio=50\%. The advantage of a HSLT light-tree becomes more evident as the group size grows, while it becomes less significant when the light splitting ratio grows.

(2) As plotted in Fig.~\ref{fig: number of destinations}(a) when group member Ratio = 25\%, HSLT can span up to about 1.5 destinations (1.5/6 = 25\%) more than R2S and one destination (1/6 = 17\%) more than MO if there is no splitter. And as shown in Fig.~\ref{fig: number of destinations}(b) where Ratio = 50\%, up to 2.5 (2.5/13 = 19\%) and 2 (2/13 = 15\%) more destinations could be included by a HSLT light-tree compared to R2S and MO, respectively. These results indicate that as the group size grows, more destinations could be spanned in a HSLT light-tree than that in the light-trees computed by the other two algorithms. The two figures also show that the advantage of a HSLT light-tree becomes significant as the number of MC nodes decreases.

(3) From the curves plotted in Fig.~\ref{fig: total cost}, HSLT has a slightly lower total cost than MO algorithm, which was supposed to achieve the best cost of all the current algorithms.

The above observations can be explained as follows. When the group size is big, in the modified graph more destinations could be joined to a HSLT light-tree via a hypo-shortest path while the other two algorithms cannot because they use and only use the shortest path in the original topology. This is why the HSLT algorithm has a bigger capacity to include more destinations in a light-tree than the others. However, as the ratio of MC nodes increases, when building a light-tree, the size of set $MC\_SET$ grows while the size of set $MI\_SET$ reduces, thereby leading to the release of light splitting constraint. Hence, the numbers of destinations included in these three kinds of light-trees grow until they reach the same value. Hence, the HSLT light-tree is more advantageous in sparse light splitting WDM networks. As far as the link stress, the HSLT algorithm is able to achieve the best performance; since the more destinations could be spanned in a light-tree, the fewer light-trees will be required by a multicast session.

Secondly, we investigate the network throughput and efficiency of wavelength usage when the number of wavelengths supported in a fiber link is set to $W = 20$, group size $K+1$ is a randomly variable uniformly generated from $[3,N]$, and the First-Fit~\cite{yzZhu2000} wavelengths assignment algorithm is employed. In Fig.~\ref{fig: Throughput and Efficiency}(a), with the help of HSLT light-trees, up to 1.6 more multicast sessions could be accepted in the network, which indicates a 9\% improvement of network capacity. And just because more wavelengths in fiber links could be used to establish some extra multicast sessions, the HSLT reaches high to 85\% wavelengths usage as shown in Fig.~\ref{fig: Throughput and Efficiency}(b).

\section{Conclusion}
\label{sec:Conclusion}
Sparse light splitting constraint and the absence of wavelength converters complicate multicast routing in WDM mesh networks. To solve the light-trees construction problem for multicast routing, the Hypo-Steiner Heuristic is introduced to overcome the drawback of Member-Only algorithm~\cite{xjzhang2000}, which only makes use of the shortest path in the original topology to connect a destination to a light-tree under construction. The Hypo-Steiner Heuristic always deletes the non-leaf MI nodes in the light-tree under construction from the topology, and searches a hypo-shortest path in the modified topology. At each step, after the pruning operation, there is no non-leaf MI node any more in the modified topology, thus the shortest path found (a hypo shortest path in the original topology) is always constraints-satisfied. With the Hypo-Steiner Heuristic, more possible paths could be found to connect a destination to the light-tree under construction. That is why more destinations could be spanned in a HSLT light-tree than Member-Only. This consequently results in the decrease of link stress, which equals to the number of light-trees required for a multicast session. Regarding the wavelength channel cost consumed by per multicast session, the total cost of light-trees built by Hypo-Steiner Heuristic is proved to be upper bounded to $N(N-1)/2$, where $N$ is the number of nodes in the network. Extended simulations are also implemented. The numeric results not only validate the analysis proof but also show that the Hypo-Steiner Heuristic is more advantageous than the existing light-tree construction algorithms in sparse light splitting all-optical WDM mesh Networks.

\end{document}